\documentclass[11pt,a4paper]{article}
\usepackage{amsmath, amsthm, amscd, amsfonts, amssymb, graphicx, color}
\usepackage[utf8]{inputenc}
\usepackage{graphicx}
\usepackage{hyperref}
\usepackage[titlenumbered,ruled,noend,algosection]{algorithm2e}

\newtheorem{theorem}{Theorem}[section]
\newtheorem{corollary}[theorem]{Corollary}
\newtheorem{lemma}[theorem]{Lemma}

\begin{document}
\title{Packing Rotating Segments}
\author{Ali Gholami Rudi\thanks{
{Department of Electrical and Computer Engineering},
{Bobol Noshirvani University of Technology}, {Babol, Mazandaran, Iran}.
Email: {\tt gholamirudi@nit.ac.ir}.}}
\date{}
\maketitle
\begin{abstract}
We show that the following variant of labeling rotating
maps is NP-hard, and present a polynomial approximation
scheme for solving it.
The input is a set of feature points on a map, to each
of which a vertical bar of zero width is assigned.
The goal is to choose the largest subsets of the bars such
that when the map is rotated and the labels remain vertical,
none of the bars intersect.
We extend this algorithm to the general case where labels are arbitrary objects.
\end{abstract}
\\

\section{Introduction}
We study the following problem for a set of points on the plane.
To each of these points a vertical segment is assigned.  The goal is to
place the maximum possible number of these segments on
the plane such that: i) each segment intersects its
corresponding point (the point is the anchor of the segment),
ii) when the plane is rotated, each segment is rotated in the
reverse direction around its anchor point to remain vertical,
iii) during rotation of the plane, no two segments intersect.

Placing as many labels as possible on a map (known as map labeling) is
a classical optimization problem in cartography and graph drawing \cite{formann91}.
For static maps, the problem of placing labels on a map can be stated
as an instance of geometric independent set problem (sometimes also
called packing for fixed geometric objects): given a set of geometric
objects, the goal is to find its largest non-intersecting subset.
In the weighted version, each object also has a weight and the
goal is to find a non-intersecting subset with the maximum possible weight.

A geometric intersection graph, with a vertex for each object and
an edge between intersecting objects, converts this geometric problem to
the classical maximum independent set for graphs, which is proved to be
NP-hard and difficult to approximate even for a factor of $n^{1-\epsilon}$,
where $n$ is the number of vertices and $\epsilon$ is any non-zero positive constant \cite{hastad96}.
Although the geometric version remains NP-hard even for unit disks \cite{fowler81},
it is easier to approximate and several polynomial-time approximation
schemes (PTAS) have been presented for this problem \cite{hochbaum85,agarwal98,erlebach05,chan03,chan12}.
Note that a PTAS finds a $(1 - \epsilon)$-approximate solution in time $o(n^{f(\epsilon)})$,
for a parameter $\epsilon > 0$ and some function $f$ independent of $n$.

Maps may be dynamic, and allow zooming, panning, or rotation,
as recent technology has made prevalent.  Most work on labeling dynamic
maps consider zooming and panning operations \cite{been06}, and few
results have been published for labeling rotating maps.
Gemsa et al.~\cite{gemsa11} where the first to study this problem.
They assumed the model presented by Been et al.~\cite{been06} for zoomable maps,
to define the consistency of a rotating map.
For the $k$R-model, in which each label may disappear at most $k$ times during
rotation, they showed that labeling rotating maps is NP-hard, even for
unit-height labels, when the goal is to maximise the total duration in
which labels are visible without intersecting other labels.
For unit-height labels, they also presented a $1/4$-approximation algorithm
and a PTAS, with the assumption that the number of anchor points contained
in any rectangle is bounded by a constant multiplied by its area,
each label may intersect a constant number of other labels,
and the aspect ratio of the labels is bounded (the
first two may not hold in real world maps).
They~\cite{gemsa16} later extended their results by presenting
heuristic algorithms, and also an integer linear programming (ILP)-based
solution for labeling rotating maps under the same assumptions.
They also experimentally evaluated these algorithms.
The size of their ILP modeling of the problem was later improved
by Cano et al.~\cite{cano17}.

Yokosuka and Imai~\cite{yokosuka13} solved the problem of maximising
the size of labels for rotating maps.  Although this problem is
NP-hard for static maps, they presented an exact $O(n \log n)$-time
algorithm for the case where the anchor points can be inside the
labels, and also when the labels are of unit height and the points
can be on the boundary of the labels.
Gemsa et al.~\cite{gemsa13} also studied a trajectory-based labeling
problem, when the trajectory of the viewport of the map is specified
as an input.

In this paper, we study a variant of the general problem of
labeling rotating maps, where the labels have zero width and
the goal is to find the largest subset of the labels that
remain disjoint during rotation.
Unlike Gemsa et al.~\cite{gemsa11}, we do not make simplifying
assumptions about the distribution of the labels:
a label may intersect any number of other labels, and the
number of feature points in any rectangle may not be proportional
to its area.
We present a PTAS for this problem, and then extend our results
to the general case, where the labels can be arbitrary objects.

This paper is organised as follows.
We first prove that labeling rotating segments is NP-hard in
Section~\ref{spre}.  Then, in Section~\ref{spck}, we present a
PTAS for this problem, after reviewing some of the PTAS presented
for the geometric independent set problem.
We finish this section by extending the approximation scheme
to arbitrary objects.
Finally, we conclude this paper in Section~\ref{scon}.

\section{Notation and Preliminary Results}
\label{spre}
Let $P = \{ p_1, p_2, ..., p_n \}$ be a set of points, and
$\ell_i$ be the length of the vertical segment corresponding to $p_i$.
A labeling $\phi$ for $P$, assigns a vertical segment
to some of the points in $P$.  If a segment is assigned to
point $p_i$ in labeling $\phi$, we say $p_i$ is included in $\phi$,
or equivalently, $p_i \in \phi$.
The notation $\phi({p_i})$ denotes the segment assigned to $p_i$,
and $|\phi|$ indicates the size of $\phi$.
Note that the length of the segment assigned to $p_i$ is $\ell_i$.

If segment $s$ is assigned to $p_i$ in $\phi$,
$s$ should intersect $p_i$.  The point at which $s$
and $p_i$ intersect is the \emph{anchor} point of $s$;
alternatively, we say $s$ is anchored at $p_i$.
When the plane is rotated, $s$ is rotated in the reverse
direction around $p_i$ to remain vertical.
In 1-position (1P) model, the anchor point of a segment
should be its bottom end point.
In 2-position (2P) model, either the top or the bottom
end points of a segment can be its anchor point.
In fixed-position (FP) model, the anchor point of each
segment is fixed (but different segments may be anchored
at different positions).
In the sliding model, any point on the segment can be its
anchor point.

A labeling is \emph{proper}, if its assigned segments do not intersect
during the rotation of the plane.
In the Maximum Rotating Independent Set (MRIS) problem for vertical
segments, the goal is to find the largest proper labeling.
Instead of rotating the plane and keeping visible segments vertical,
we can equivalently fix the plane and rotate all visible labels in
the reverse direction.  This is what we do in the rest of this paper.

In Theorem~\ref{tred}, we show that MRIS for vertical segments is
NP-hard by a reduction from the Geometric Maximum Independent Set (GMIS)
for unit disks, which is already proved to be NP-hard~\cite{fowler81}.

\begin{theorem}
\label{tred}
MRIS for a set of segments in 1P model is NP-hard.
\end{theorem}
\begin{proof}
We reduce any instance of GMIS for unit disks to an instance of MRIS for segments.
Let $D$ be a set of $n$ unit disks on the plane and
let $P$ be the centre of these disks.
Also, let $\ell_i$ be $2$ for $1 \le i \le n$.
We show that from every non-intersecting subset of disks in $D$
we can obtain a proper labeling of $P$ with the same size.

Let $D'$ be a non-intersecting subset of $D$, and let $P'$ be their centres.
Since the disks in $D'$ are non-intersecting,
the distance between any pair of points in $P'$ is at least 2.
Let $\phi$ be the labeling of size $|D'|$ that assigns a
segment of length $\ell_i$, anchored at its bottom end point,
to each point $p_i$ of $P'$.  These segments cannot intersect during
rotation: the segments are always parallel, and since the distance
of their anchors is at least 2, they do not intersect.
This implies that $\phi$ is proper.

For the other direction, let $\phi$ be a proper labeling of $P$.
Let $D'$ be the set of disks, whose centres are in $\phi$.
Since $\phi$ is a proper labeling in 1P model,
the distance between any pair of points in $P'$ is at least two.
This implies that the disks corresponding to $P'$ are non-intersecting.
\end{proof}

In Lemma~\ref{l2to1}, we show that labeling in the 2P model is as
difficult as labeling in 1P model.  Note that a labeling in
the latter is also a labeling in the former.

\begin{lemma}
\label{l2to1}
Let $\phi$ be a proper labeling of a set of points $P$ in the 2P model.
If all segments in $\phi$ are changed to be anchored at their bottom end point,
the resulting labeling is also proper.
\end{lemma}
\begin{proof}
Let $\phi'$ be the labeling obtained by changing the
anchor points of all segments assigned in $\phi$.
If $\phi'$ is not proper, there exists at least a pair of points $p_i$ and $p_j$,
such that $\phi'(p_i)$ and $\phi'(p_j)$ intersect during rotation.
Without any loss of generality, suppose $\ell_i \ge \ell_j$.
Therefore, the distance between $p_i$ and $p_j$ is at most $\ell_i$,
implying that, at some point during rotation $\phi(p_i)$ intersects
$p_j$ (and thus $\phi(p_j)$), contradicting the assumption that $\phi$
is a proper labeling.
\end{proof}

For the sliding model, in Lemma \ref{lsld} we show that there is an
optimal labeling, in which all labels are anchored at their midpoint.

\begin{lemma}
\label{lsld}
Let $\phi$ be a proper labeling of a set of points $P$ in the sliding model.
If all segments assigned in $\phi$ are changed to be
anchored at their midpoint, the resulting labeling is also proper.
\end{lemma}
\begin{proof}
Let $\phi'$ be the new labeling.
If $\phi'$ is not a proper labeling, there exists points
$p$ and $q$ such that $\phi'(p)$ and $\phi'(q)$ intersect
during rotation, which implies that the distance between
$p$ and $q$ is at most $(\ell_p + \ell_q) / 2$,
where $\ell_p$ and $\ell_q$ are the lengths of the segments
assigned to $p$ and $q$, respectively.
Let $\phi(p) = p'p''$ and $\phi(q) = q'q''$, in which
$p'$ and $q'$ are the top end points of these segments.
Obviously, $|pp'| + |pp''|$ is equal to $\ell_p$ and
$|qq'| + |qq''|$ is equal to $\ell_q$.
Duration rotation, when $q'$ is on segment $pq$, we have
$|qq'| + |pp''| < |pq|$; otherwise the segments intersect.
Similarly, when $q''$ is on $pq$, we have $|qq''| + |pp'| < |pq|$.
This, however, implies that $|qq'| + |qq''| + |pp'| + |pp''| < 2 \cdot |pq|$,
or equivalently $\ell_p + \ell_q  < 2 |pq|$.
This yields $|pq| > (\ell_p + \ell_q) / 2$, and a contradiction.
Therefore, $\phi'$ is also a proper labeling.
\end{proof}

In the next section, we study MRIS problem for segments in 1P model,
but our results also apply to 2P model by Lemma~\ref{l2to1}.

\section{A PTAS for MRIS}
\label{spck}
For a set of points $P = \{p_1, p_2, ..., p_n\}$,
let $D_i$ denote the disk centred at $p_i$ with radius $\ell_i$.
In a proper labeling $\phi$ in 1P model,
if points $p_i$ and $p_j$ are both present in $\phi$,
their corresponding disks, $D_i$ and $D_j$, may intersect, but
cannot contain the centre of the other.
This is what makes MRIS for segments in 1P model different from
GMIS, in which the objects in the output should be disjoint.
In this section, we first review some of the PTAS presented
for GMIS, and adapt one of them for our problem.

Note that transforming an instance of MRIS to GMIS, based
on the idea used in Theorem~\ref{tred}, does not work,
since the length of the segments (the radius of the disks)
are not equal.  To see this, consider two disks $D_1$
and $D_2$ of radius 2 and 20, respectively, in an instance of MRIS.
To obtain an equivalent instance for the GMIS, we replace each
disk with a disk of half its radius, as in Theorem~\ref{tred}.
Therefore, we have two disks $D'_1$ and $D'_2$,
corresponding to $D_1$ and $D_2$, of radius 1 and 10, respectively.
$D'_1$ and $D'_2$ may be a solution in the GMIS instance,
but their corresponding disks may not be a solution in the MRIS
instance ($D'_1$ and $D'_2$ may be disjoint but $D_2$ may contain
the centre of $D_1$).

\subsection{Geometric Maximum Independent Set}

In what follows, we review some of the PTAS presented for GMIS.
Hochbaum and Maass~\cite{hochbaum85} presented a PTAS for packing
$n$ unit disks or squares on the plane in $n^{O(1 / {\epsilon^2})}$ time.
The algorithm places a grid on the plane and by removing the objects that
intersect the boundary of grid cells, solves the problem for each cell
independently using brute force.
To find the exact solution, this process is repeated
after shifting the grid a polynomial number of times in $\epsilon$.
Agarwal et al.~\cite{agarwal98} improved this algorithm using dynamic
programming to work for fat objects (informally,
objects with bounded aspect ratio) of similar size to
$n^{O(1/ \epsilon^{d-1})}$, in which $d$ is the number of dimensions.

Erlebach et al.~\cite{erlebach05} extended Hochbaum and Maass' algorithm
using a multi-level grid to handle arbitrary-sized but fat
objects.  To find an optimal solution for each grid cell
and for every disjoint subset the objects that intersect
its boundary, the problem is recursively solved for lower
grid levels using dynamic programming.
Similar to Hochbaum and Maass~\cite{hochbaum85},
they shift the grids to obtain an exact solution.
Chan~\cite{chan03} presented a similar algorithm for fat objects using
a quadtree instead of multi-level grids, improving the time
complexity to $n^{O(1/ \epsilon^{d-1})}$.
In the same paper, Chan~\cite{chan03} also presented a divide-and-conquer
algorithm based on the geometric separator theorem~\cite{smith98},
in time $n^{1/\epsilon^d}$ in the plane for unweighted and fat objects,
improving the space complexity of the previous algorithm.
More recently, Chan and Har-Peled~\cite{chan12} presented an PTAS
based on local search with the time complexity $n^{1/\epsilon^d}$.
They also presented a constant-factor approximation algorithm
based on linear programming relaxation.

In most of these results, input objects are assumed to be fat;
``a set of objects are fat if for every axis-aligned hypercube of
side length $r$, we can find a constant number of points such that every
object that intersects the hypercube and has diameter at least $r$
contains one of the chosen points \cite{chan12}.''
For non-fat objects, approximating independent sets is more
difficult, and, for instance for arbitrary rectangles,
only recently a quasi-PTAS has been presented~\cite{adamaszek13}.

\subsection{An Algorithm for Packing Rotating Segments}
\label{sseg}

We adapt Chan's \cite{chan03} shifted quadtree algorithm for solving MRIS for
segments in 1P model.  The main reason for preferring this algorithm
to other PTAS for GMIS, is its lower time complexity.
The algorithm presented by Hochbaum and Maass \cite{hochbaum85} is simpler
but cannot handle arbitrary sized segments in MRIS.
To make this section mostly self-contained, we repeat the necessary definitions
and proofs from \cite{chan03}, and try to simplify them where possible.

The algorithm presented by Chan \cite{chan03}, which uses
the shifting-quadtree technique, assumes fat input objects.
Also, the objects in the output of GMIS should be disjoint.
In MRIS input objects are rotating segments (segments are not fat).
We can consider the disks that result from the rotation of
these segments, but then, the objects in the output of MRIS
may not be disjoint.  Therefore, the results of \cite{chan03}
do not apply to MRIS.  We modify that algorithm for solving MRIS.

For simplicity, we first map (using scaling and translation) all disks
$D = \{D_1, D_2, ..., D_n\}$ to fit inside a unit square with its lower
left corner at the origin, and store them in a quadtree.
Let $r_i$ be the radius of disk $D_i$ after this mapping.
A quadtree cell at depth $d$ has side length $2^{-d}$.
Two disks are \emph{centre-disjoint}, if none contains the centre of the other.
A disk of radius $\ell$ is \emph{$k$-aligned}, if it is inside a quadtree
cell of size at most $k \ell$.

\begin{figure}
	\centering
	\includegraphics{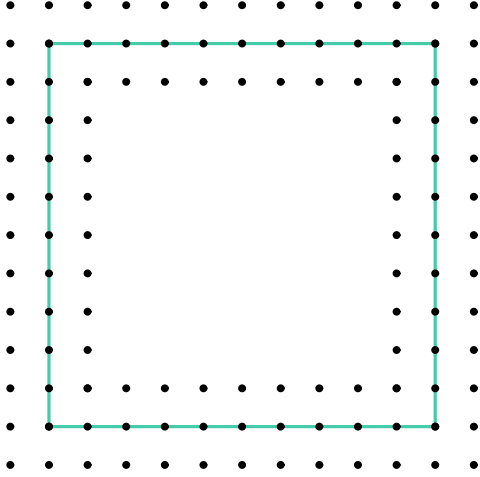}
	\caption{Points placed around the boundary of a grid cell in Lemma~\ref{lfat}}
	\label{ffat}
\end{figure}

\begin{lemma}
\label{lfat}
Let $C$ be a quadtree cell at depth $d$, and $B$ be a set of $k$-aligned disks
that intersect the boundary of $C$.
The size of any centre-disjoint subset $S$ of $B$ is bounded by $ck$
for some constant $c$.
\end{lemma}
\begin{proof}
If a $k$-aligned disk intersects $C$, its radius should be
at least $r = 2^{-d}/k$, based on the definition of $k$-aligned objects.
Therefore, since all members of $B$ are $k$-aligned, the radius of any
of them is at least $r$.
Let $C'$ and $C''$ be squares of side length $2^{-d} - 2r$ and $2^{-d} + 2r$,
with the same centre of gravity as $C$.
Place a set $X$ of points on $C$, $C'$, and $C''$ with distance $r$,
as shown in Figure~\ref{ffat}.
Since $2^{-d}/r$ is $k$, the size of $X$ no greater than $12k$.
Any disk of radius at least $r$ that intersects the boundary of $C$
contains at least one point from $X$.
On the other hand, for any point $p$ on the plane,
the maximum number of centre-disjoint disks that can contain $p$ is 6.
Therefore, the number of points in any centre-disjoint subset of $B$
is at most $72k$.
\end{proof}

Let $C$ be a quadtree cell, $B$ be a set of centre-disjoint disks
intersecting its boundary, and $I$ be a set disks inside $C$.
Let $\textrm{MRIS}(C, B, I)$ denote the maximum size of a
centre-disjoint subset of $I$ such that its union with $B$ is also
centre-disjoint.

\begin{lemma}
\label{lquad}
Suppose all disks in $D$ are stored in a quadtree and $k$-aligned.
Let $C$ be a quadtree cell, $B$ be a set of disks intersecting
the boundary of $C$, and $I$ be the set of disks completely inside $C$.
The value of $\textrm{MRIS}(C, B, I)$ can be computed
from the value of $\textrm{MRIS}$ for the children of $C$
in time $n^{O(k)}$.
\end{lemma}
\begin{proof}
Let $C_i$ for $1 \le i \le 4$ denote the child cells
of $C$ in the quadtree.
For a set of disks $X$, let $C^\textit{in}(X)$ and $C^\textit{on}(X)$ 
denote the subset of $X$ completely inside and the subset intersecting
the boundary of $C$, respectively.
Let $B'$ be the subset of $I$ that intersect the boundary
of the children of $C$.
For a centre-disjoint subset $J$ of $B'$, whose
members are also centre-disjoint from the members of $B$,
the value $\textrm{MRIS}_J(C, B, I)$,
denoting the maximum size of a centre-disjoint subset of $I \setminus B$
that includes $J$ equals:
\[
\textrm{MRIS}_J(C, B, I) = \sum_{C_i} \textrm{MRIS}(C_i, C_i^\textit{on}(B \cup J), C_i^\textit{in}(I)) + |J|
\]
To compute the value of $\textrm{MRIS}(C, B, I)$,
we find the maximum value of $\textrm{MRIS}_J(C, B, I)$ for
every centre-disjoint subset $J$ of $B'$.
Using an argument similar to Lemma~\ref{lfat}, we can show that the
size of any centre-disjoint subset of $B'$ is bounded by $O(k)$.
Therefore, we can compare the value of $\textrm{MRIS}_J(C, B, I)$
for every such subset (there are at most $n^{O(k)}$ such subsets)
to find the value of $\textrm{MRIS}(C, B, I)$ in time $n^{O(k)}$.
\end{proof}

\begin{theorem}
\label{tseg}
A $(1+\epsilon)$-approximate solution to MRIS for segments in 1P model
can be computed in time $n^{O(1/\epsilon)}$, for any real constant $\epsilon$,
where $0 < \epsilon < 1$.
\end{theorem}
\begin{proof}
Let $P$ and $D$ be defined as in the beginning of this section.
Let $k = 1/\epsilon$.
We can store $D$ in a compressed quadtree (a quadtree in which nodes
with only one non-empty child cell are merged, resulting in $O(n)$
nodes), and modify Lemma~\ref{lquad} to consider merged nodes.
Let $\overline{C}$ be the root cell of this quadtree.
We compute the value of $\textrm{MRIS}(C, B, I)$ for every possible
quadtree cell $C$, and inputs $B$ and $I$ in a bottom-up manner.
\begin{enumerate}
\item
For leaves, $\textrm{MRIS}(C, B, I)$ can be computed in $O(1)$ ($I$
has only one member).
\item
For any other cell $C$, $I$ is always $C^\textrm{in}(D)$, and
$B$ is always a subset of $C^\textrm{on}(D)$ (there are $n^{O(k)}$
subsets); computing $\textrm{MRIS}(C, B, I)$ for every such
input takes $n^{O(k)}$ time by Lemma~\ref{lquad}.
\end{enumerate}
Therefore, we can find the exact value of $\textrm{MRIS}(\overline{C}, \emptyset, D)$,
by computing $\textrm{MRIS}$ for every node of the quadtree
recursively in time $n^{O(1/\epsilon)}$,
supposing every disk is $k$-aligned.

Using the shifting technique of \cite{chan03} (also \cite{hochbaum85} and
\cite{erlebach05}), we can translate the disks $k$ times,
such that in one of these translations at least $(1 + \epsilon)$ of the
disks in an optimal solution to MRIS for $D$ are $k$-aligned;
computing $\textrm{MRIS}(\overline{C}, \emptyset, D)$ after
each such translation, and taking their maximum value,
achieves the desired approximation factor.
\end{proof}

\subsection{Packing Arbitrary Rotating Objects}
\label{sobj}
The algorithm of Section~\ref{sseg} can be extended to work for a
combination of vertical and horizontal segments (or of any orientation),
or even for arbitrary objects in FP model.
For arbitrary objects, we similarly denote with $D_i$ the disk
that results by rotating the $i$-th object around its anchor point.
Unlike vertical segments in 1P model, two objects may intersect during
rotation, even if their corresponding disks are centre-disjoint.
To see this, consider a horizontal segment of unit length,
anchored at its left endpoint at the origin, and a vertical segment
of unit length, anchored at its bottom endpoint at $(1, -.5)$.

We modify the results of Section~\ref{sseg} as follows.
Instead of finding a centre-disjoint subset of a set of
disks, our goal is finding a subset of disks, such
that there exists a proper labeling that includes all of the
objects that correspond to them.
Lemma~\ref{lfat} can therefore be modified as follows (Lemma~\ref{lany}).

\begin{lemma}
\label{lany}
Let $C$ be a quadtree cell at depth $d$, $B$ be a set of $k$-aligned
disks that intersect the boundary of $C$, and $O$ be the set
of objects that correspond to the members of $B$.
The size of any proper labeling $\phi$ of $O$ is bounded by $ck$
for some constant $c$.
\end{lemma}
\begin{proof}
Since the anchor of each object is inside that object,
only centre-disjoint objects can appear in $\phi$ (if an
objects intersects the anchor of another object during
rotation, they certainly intersect).
Lemma~\ref{lfat} shows that the size of any centre disjoint
subset of $B$ is bounded by $ck$ for some constant $k$.
Thus, the size of $\phi$ cannot be any greater.
This implies the required upper bound.
\end{proof}

Using Lemma~\ref{lany} in Lemma~\ref{lquad} and Theorem~\ref{tseg},
they imply the following Corollary, after slight modifications.

\begin{corollary}
A $(1+\epsilon)$-approximate solution to MRIS for arbitrary objects in FP model
can be computed in time $n^{O(1/\epsilon)}$, for any real constant $\epsilon$,
where $0 < \epsilon < 1$.
\end{corollary}

\section{Concluding Remarks}
\label{scon}
The algorithms presented in Section~\ref{spck} can be extended to solve
MRIS in $\mathbb{R}^d$ for $d > 2$, where a $d$-dimensional map can be
rotated in any direction, with the time complexity $n^{O(1/{\epsilon^{d-1}})}$.


\end{document}